\documentclass[preprint]{sigplanconf}

\usepackage{amsmath}
\usepackage{amsthm}
\usepackage{amssymb}
\usepackage{graphicx}

\usepackage{algorithmic}
\usepackage{algorithm}
\usepackage{url}

\newtheorem{lem}{Lemma}

\begin{document}

\copyrightyear{2010}
\preprintfooter{Technical report, University of Vienna}
\title{Work-stealing for mixed-mode parallelism by
deterministic team-building\thanks{The 
research leading to these results has received funding from the 
European Union Seventh Framework Programme (FP7/2007-2013)
under grant agreement no.\ 248481 (PEPPHER Project, \protect\url{www.peppher.eu}).}}
\authorinfo{Martin Wimmer, Jesper Larsson Tr\"aff}
{Department of Scientific Computing, Faculty of Computer Science \\
University of Vienna/Wien\\
Nordbergstrasse 15/3C, A-1090 Wien, Austria}
{\{wimmer,traff\}@par.univie.ac.at}

\maketitle

\begin{abstract}
We show how to extend classical work-stealing to deal also with
\emph{data parallel tasks} that can \emph{require} any number of
threads $r\geq 1$ for their execution. As threads become idle they
attempt to join a \emph{team} of threads designated for a task
requiring $r>1$ threads for its execution. Team building is done
following a deterministic pattern involving $\log p$ possibly
randomized steal attempts where $p$ is the number of started hardware
threads.  Deterministic work-stealing often exhibits good locality
properties that are desirable to preserve.  Threads attempting to join
the team for a task requiring a large team may help smaller teams
instead of waiting for the large team to form. We explain in detail
the so introduced idea of \emph{work-stealing with deterministic
team-building} which in a natural way generalizes classical
work-stealing. The implementation is done with standard lock-free data
structures, in addition to which only a single extra compare-and-swap
(CAS) operation per thread is required as a team is being built. Once
formed, teams can stay to process further tasks requiring the same (or
smaller) number of threads; this requires no further coordination.  In
the \emph{degenerate case}, where all tasks require only a single
thread, the implementation coincides with a (deterministic)
work-stealing implementation, has no extra overhead, and therefore
similar theoretical properties. We demonstrate correctness of the
generalized work-stealing algorithm by arguing for deadlock freedom
and completeness (all tasks will eventually be executed, regardless of
their resource requirement $r\leq p$), discuss its load-balancing,
task execution order and memory-consumption properties, and discuss a
number of algorithmic and implementation variations that can be
considered. A prototype C++ implementation of the generalized
work-stealing algorithm has been given and is briefly
described. Building on this, a serious, well-known contender for a
best \emph{parallel Quicksort} algorithm has been implemented, which
naturally relies on both task and data parallelism.  On an 8-core
Intel Nehalem system, a 16-core AMD Opteron system, a 16-core Sun T2+
system supporting up to 128 hardware threads, and a 32-core Intel
Nehalem EX system we compare our implementation of the published
Quicksort algorithm using fork-join parallelism to a mixed-mode
parallel implementation with a data parallel partitioning step using
our deterministic team-building work-stealer. Results are consistently
better. often by a significant fraction.  For instance, sorting
$2^{27}-1$ randomly generated integers we could improve the speed-up
from 5.1 to 8.7 on the large 32-core Intel system, on this system
being consistently better than the tuned, task-parallel Cilk++ system.
\end{abstract}

\section{Introduction}

Work-stealing is a now classical, efficient strategy for dynamically
scheduling parallel work-loads of independent, sequential tasks on
shared-memory systems with possibly varying number of available
processing resources~\cite{BlumofeLeiserson99,AroraBlumofePlaxton01}.
With work-stealing, the sequential tasks of a DAG-structured
computation are executed by the available independent hardware
threads. Ready (and newly spawned) tasks are kept in local queues, and
only when a thread locally runs out of tasks does it attempt to steal
work (tasks) from other threads. Despite its localized nature with no
global synchronization, it is nevertheless often possible to prove
good time bounds and thread and memory/cache utilization for
work-stealing based schedulers~\cite{BlumofeLeiserson99,AroraBlumofePlaxton01}.
Work-stealing is used as the basis in
Cilk~\cite{BlumofeJoergKuszmaulLeisersonRandallZhou96}, Intel's
TBB~\cite{KukanovVoss07}, and many other task-parallel programming
systems.  Efficient implementation of work-stealing relies heavily on
lock- and/or wait-free data structures~\cite{HerlihyShavit08}.

In the dynamic task-based programming models that fit well with
work-stealing, data-parallel loops are typically handled by
recursively breaking the loop into chunks that are then handled
sequentially by the available hardware threads. Work-stealing provides
no means of ensuring simultaneous scheduling of the tasks responsible
for such pieces, and no control over where (and when) the pieces are
eventually executed. Thus, data-parallel tasks with dependencies are
not well suited to work-stealing. This limitation has often been
addressed and frameworks which allow communicating tasks have been
proposed, see e.g.~\cite{Dummler07}. Phasers as known from
Habanero~\cite{Shirako08} allow loose synchronization of
single-threaded tasks.  Conversely, sometimes parallel algorithms
(e.g.\ the Quicksort algorithm that will be explained in
Section~\ref{sec:quicksort}) are conveniently formulated as a sequence
of parallel (recursive) steps followed by a bunch of independent
sequential work (perhaps followed by merging of results again done in
parallel). Such computations are likewise not easily executed by
work-stealing schedulers.

Mapping data-parallel algorithms to task-based programming models has
drawbacks, which can be resolved by mixing task- and data-parallel
programming. Algorithms naturally requiring task and data parallelism
and the benefit that can be expected from \emph{mixed data and task
parallel programs} are discussed
in~\cite{ChakrabartiDemmelYelick97}. Centralized scheduling methods
for handling \emph{mixed data and task parallel programs} were
discussed for instance already in~\cite{ChakrabartiDemmelYelick97},
see
also~\cite{BoudetDesprezSuter03,CrowlCrovellaLeBlancScott94,DesprezSuter04,Dutot09,RadulescuGemund01,Rauber99,Rauber07}.

The model we consider here is DAG-structured computations with
dynamically spawned, non-malleable tasks with fixed thread
requirements, that is, each new task must be executed by some number
of threads determined at spawn time. The problem here is how to gather
the threads that will eventually execute data-parallel tasks requiring
more than a single thread, avoid unnecessary idle times in the
process, make sure that such gathered threads can be activated
together, and that a convenient virtual numbering of the threads is
available, such that the co-scheduled tasks have a means of
identifying and communicating with each other.

A dynamic, greedy approach like work-stealing might be able to
circumvent or alleviate some of these problems, and still provide a
(provably) efficient utilization of resources.  In this paper we propose
to extend classical work-stealing in this direction. As far
as we are aware, such a generalization of work-stealing has not been
given before.

This model does at first glance not fit naturally with work-stealing
in which coordination is done locally by thieves running out of work
and no centralized resources are available for co-scheduling data
parallel tasks over some set of available hardware threads. The
contribution of this paper is to give a natural extension of
work-stealing that allows for execution of such mixed data and task
parallel programs. The extension is called \emph{deterministic team
building}. When a thread runs out of tasks in its local queue, it
tries to help other threads to execute a data-parallel task requiring
more than a single thread for its execution thereby forming a
\emph{team}. Coordination is thus, like in work-stealing, done by the
\emph{thieves} and not coordinated from ``the top'' by the threads
having the data-parallel tasks in their local queues. The overhead for
forming a new team is a single extra atomic compare-and-swap (CAS)
instruction per thread joining a team. In order to avoid idle times of
threads waiting for large teams to form, threads wanting to join a
large team can help threads with tasks requiring fewer threads. In
order to ensure that sufficiently large teams can be formed fast, and
that a consecutive thread numbering within teams can be computed,
work-stealing and team-building is done following a deterministic
pattern. Each thread becoming idle and wanting to steal has $\log p$
unique partners ($p$ being the number of hardware threads) from which
it attempts to steal work respectively build teams.

This paper concentrates on presenting the algorithmic idea of
work-stealing with deterministic team-building. A basic implementation
has been given in C++, and we have used this to give a very natural
implementation of a parallel Quicksort
algorithm~\cite{TsigasZhang03} with exactly the properties of having
dependent, data parallel computations (of decreasing granularities)
mixed with sequential sorting of smaller chunks of the input.

The work described here was partly motivated by the European FP7
project PEPPHER (for ``PErformance Portability and Programmability for
Heterogeneous many-core aRchitectures'', see \url{www.peppher.eu})
that develops a framework for enhancing performance portability of
applications that consist of component-tasks that may already have
been parallelized and make explicit requirements for specific
(processor) resources (with complimentary guarantees of staying within
the requested limits). Such component-tasks are in this context typically
non-malleable. Among other issues, PEPPHER investigates scheduling
strategies and software for such situations.

\section{Standard Work-stealing}

To set the stage for the description of \emph{deterministic team-building}, 
the standard work-stealing framework upon which our algorithm is built is
shown as Algorithms~\ref{alg:mainloop}, \ref{alg:gettask}, \ref{alg:stealtask}
and~\ref{alg:popappend}.

From now on the number of hardware threads is denoted by $p$.
Individual threads maintain tasks in local, double-ended queues
denoted by $Q$ with the usual operations popBottom(), pushBottom(),
popTop() and isEmpty(). These queues are assumed to be implemented in
a lock/wait-free
manner~\cite{AroraBlumofePlaxton01,HerlihyShavit08}. The main loop
(Algorithm~\ref{alg:mainloop}) terminates when all local queues are
empty and no tasks are running. Termination detection details are not
shown. Tasks are run by invoking their run() method.  Running tasks
can spawn new tasks, and are responsible for putting these onto the bottom of
the local queues by corresponding pushBottom() operations.

\begin{algorithm}
\begin{algorithmic}[1]
\caption{Basic local work-stealing loop}\label{alg:mainloop}
\WHILE {(task $\gets$ getTask()) $\neq\perp$}
\STATE\COMMENT{Get a new task and run it (eventually spawning new tasks in the process)!}
	\STATE task.run()
\ENDWHILE
\end{algorithmic}
\end{algorithm}

\begin{algorithm}
\begin{algorithmic}[1]
\caption{The getTask() procedure}\label{alg:gettask}
\REPEAT
\IF[Local queue empty]{$Q$.isEmpty()}
	\STATE stealTasks()
\ENDIF
        \STATE task $\gets$ $Q$.popBottom()
	\COMMENT{A fresh task or $\perp$ if stolen by other thread}
\UNTIL{task$\neq\perp$}
\RETURN task
\end{algorithmic}
\end{algorithm}

Procedure getTask() (Algorithm~\ref{alg:gettask}) returns a task from the 
bottom of the local queue, or steals tasks from some other thread if the local
queue is empty.

\begin{algorithm}
\begin{algorithmic}[1]
\caption{The stealTasks() procedure}\label{alg:stealtask}
\STATE $v \gets$ random() $\bmod p$
\COMMENT{Choose random victim}
\STATE $T\gets \min(v.Q.\mbox{size}()/2,\mbox{MAX\_STEAL})$)
\COMMENT{Attempt to transfer $T$ tasks from top of $v.Q$ to local $Q$}
\IF[At least one task stolen]{$Q$.popappend($v$,$T$) $>$ 0} 
	\STATE \COMMENT {Number of successfully stolen tasks returned}
	\RETURN
\ENDIF
\STATE\COMMENT{Unsuccessful stealing}
\STATE backoff()
\end{algorithmic}
\end{algorithm}

Tasks are stolen from a \emph{victim thread} by the stealTasks()
procedure (Algorithm~\ref{alg:stealtask}).  Instead of stealing only
one task, it is most often beneficial to steal some fraction of the
tasks of the victim's queue. This is implemented by the popappend()
procedure, which balances thief's and victim's queue by stealing half
the victim's tasks. For simplicity this is implemented by repeated
application of popTop() and pushBottom() operations. The number of
synchronization operations could be reduced by the use of more
complex, real bulk remove and append primitives. The number of tasks
to steal is a typical, tunable parameter in work-stealing schedulers
that can often significantly affect performance. In practice, the last
stolen task should not be added onto the queue in order to ensure it
cannot be stolen back. We omitted this from our algorithms for
readability reasons.  If stealing is unsuccessful the thief performs a
backoff(), the details of which can likewise affect performance (see
Section~\ref{sec:implementation}).

\begin{algorithm}
\caption{The popappend($v$,$T$) method implemented by standard queue 
operations.}\label{alg:popappend}
\begin{algorithmic}[1]
\STATE $i\gets 0$
\WHILE {$i<T$}
	\STATE task $\gets$ $v.Q$.popTop()
	\IF {task $\neq\perp$}
		\STATE $Q$.pushBottom(task)
	\ELSE
		\RETURN $i$
	\ENDIF
\STATE $i\gets i+1$
\ENDWHILE
\RETURN $T$
\end{algorithmic}
\end{algorithm}

\section{Work-stealing with deterministic team-building}

We can now extend work-stealing to cater also for mixed data and task
parallelism. In this case each newly spawned task can require a
certain, determined number of threads for its execution. This thread
requirement is denoted by $r$. In the standard work-stealing setting
$r=1$ for all tasks, whereas we want to allow for any $1\leq r\leq p$
number of required threads (requirements $r>p$ are of course
infeasible). Thread requirements are fulfilled by building
\emph{teams} of threads for tasks with $r>1$. When a team of $r$
threads has been formed for some task, the task can be executed. For
applications it is most often important that the threads of the team are
numbered consecutively, in order that a thread can identify and communicate
with the other tasks of the team. 

In the following we first present \emph{deterministic team-building}
for the case where the number of initially available threads is
a power of two, and the number of required threads for each
newly spawned task is also a power of two. We can present the
algorithm as an extension to the standard work-stealing implementation
of the previous section by appropriately modifying the procedures for
getting and stealing tasks.  As will be explained later, if thread $i$
has a task requiring a team of $r>1$ threads for execution, the team
that will eventually be built will consist of consecutively numbered
hardware threads $kr,kr+1,\ldots,i,\ldots (k+1)r-1$ for some $k$ in
the range $0 \leq k < p/r$. From this a virtual numbering of the
threads in the team from $0$ to $r-1$ can easily be computed.  We then
discuss main properties of the idea as compared to standard
work-stealing, and then finally show how the technique generalizes to
both arbitrary thread requirements and number of hardware threads.

In addition to the thread local queues $Q$ that will still be used to
hold tasks to be executed, each now with a thread requirement $r\geq
1$, each thread has a local, fixed (integer) id $I$, which is used to
deterministically determine the partner for work-stealing and
team-building attempts. For each of $\log p$ partners the id of
partner $i$, $0\leq i<\log p$ is determined by flipping the $i$th bit
of $I$. To access data structures associated with threads we need 
an array ThreadRef[] that maps thread id's into references
to the corresponding hardware threads.

Teams are coordinated by a \emph{coordinator} and threads locally
maintain a reference $c$ to their coordinator, which they use to poll
whether a task is ready to be executed, or if their registration at
the coordinator has been revoked.

The data-structure for each thread has the following members, which
may be accessed by other threads during the stealing and coordination
phases:
\begin{itemize}
\item
A unique id $I$ in the range $0 \leq I < p$.
\item
A double ended queue $Q$ containing
tasks. Local accesses always happen at the bottom, while stealing is
done from the top of the queue.
\item
A reference to the coordinator $c$ of the thread. If the thread is 
itself a coordinator, it contains a self-reference. This is always the
case when scheduling tasks with $r=1$.
\item
A registration structure $R$ that will be described below
is used for team creation at coordinating threads.
\item
A reference to a ready task that can be executed
by the current team, stored in the $c.\mbox{task}$ field. As soon as this
field is nonempty, all threads in the team are allowed to start
execution.
\item
A countdown $G$ for the ready task is provided, and is initialized to
$r-1$ ($r$ is the number of threads required for this task). Each
non-coordinator thread has to atomically decrement this field when
execution starts. As soon as the field is zero, the coordinator can be
sure that execution has started by all threads in the team, and is
then allowed to reset the $c.\mbox{task}$ field.
\end{itemize}

Each thread maintains a registration structure $R$ that is modified by a
\emph{compare-and-swap} (CAS) operation when necessary. The registration
structure is used for keeping track of a team being built for a task
currently at the bottom of the threads queue, and contains the following 
fields:

\begin{itemize}
\item 
The number of \emph{required} threads $r$ for the task at the bottom
of the queue. This is modified every time a new task is added to the
bottom of the local queue.
\item
The number of \emph{acquired} (or \emph{registered}) threads $a$,
which is the number of threads currently registered for the team. Only
threads can be registered that are required for a team of size $r$ (a
team of a certain size at a specific coordinator always consists of
the same threads due to the deterministic construction of the team, as
will be explained in the following). If a new task is added to bottom
that requires more threads, this number can stay. If it requires less
threads, we have to reset it to the number of teamed threads and
increment the new counter $N$ (see below) to ensure that no invalid
thread has registered.
\item 
The number of \emph{teamed} threads $t$ which is the number of threads
currently teamed up to work on a task. Teamed up threads are not
allowed to do any coordination work, except polling the coordinator
for work. A team is formed by the coordinator at task launch time, as
soon as all threads have registered. After execution of a task, the
coordinator may decide to either execute another task using the same
team, execute a task requiring less threads using a part of the team
(thereby freeing all other threads in the team), or disbanding a
team. In case a larger team has to be created for the next task, the
team must be disbanded, and team-building for the new task restarted.
\item A \emph{new counter} $N$ which is incremented every time the
coordinator decides to reset the number of acquired threads to the
current team size, to signal to all acquired threads that
team-building has to start over again. This happens every time a new
task with thread requirement $r'$ is pushed to the bottom of the
queue, where the previous task on the bottom requires $r>t$ threads
and the new requirement $r'$ is smaller than $r$, $r'<r$.  This is
needed to ensure that only consecutive threads can register for a
task. Each registered thread locally stores the current counter during
registration to be able to determine, whether the registration is
still active.
\end{itemize}

The full registration structure can be packed into a 64-bit integer,
and thus all fields updated by a single 64-bit CAS instruction by
assigning 16 bits to each field. For smaller numbers of hardware
threads even a 32-bit CAS suffices.

Except for the CAS required for modifying the registration structure,
and the atomic decrement required for updating the countdown $G$ for
the ready task, all other fields at the coordinator structure are only
written by the coordinator itself and therefore do not require atomic
primitives.

When a task spawns a task, the new task is pushed to the local queue
with pushBottom(). In addition, the registration structure is modified
depending on the number of threads required by the new task. If the
new number of required tasks $r$ is larger than the previous value, we
can just update this value. In case it is smaller, we have to reset
the number of acquired threads $a$ to the current team size $t$ to
ensure that we have not acquired any threads outside the boundaries of
the new task. In addition to that, we have to increment the
registration counter to notify all threads outside the current team
that they have to re-register. We do not allow for $r$ dropping below
$t$, so if the new task requires less than $t$ threads, we set $r
\gets t$.

Initially, each thread starts out with the coordinator reference
pointing to itself, $c=\mbox{ThreadRef}[I]$.

The modified getTask() procedure is shown as
Algorithm~\ref{alg:newgettask}.  If the thread has a coordinator $c
\neq \mbox{ThreadRef}[I]$ (set by previous stealTasks() attempts), it
will either start executing the coordinator's task if there is one
ready (team has been built), or help coordinating the task by polling
its partners (this is explained separately).  If the thread's local queue
is empty a modified steal attempt is executed (see
Algorithm~\ref{alg:newstealtask}).  Otherwise the thread starts
coordinating a task.

\begin{algorithm}
\caption{The modified getTask() procedure.}\label{alg:newgettask}
\begin{algorithmic}[1]
\WHILE {$G > 0$}
	\STATE \COMMENT {Make sure we do not have a previous coordinated task that has not yet been started by all threads.}
	\STATE backoff()
\ENDWHILE
\STATE task $\gets\perp$ 
\REPEAT
\IF {$c\neq \mbox{ThreadRef}[I]$}
	\STATE \COMMENT {This thread is in a team coordinated by another thread}
	\IF {$c$.task $\neq\perp$}
		\STATE \COMMENT {The coordinators task is ready, and this thread is in the team}
		\RETURN $c$.task
	\ELSE
		\STATE pollPartners($c.I$, $c.R.r$)
	\ENDIF
\ELSIF {$Q$.isEmpty()}
	\STATE stealTasks()
\ELSE
	\STATE coordinateTask()
	\IF {task $=\perp$}
		\STATE backoff()
	\ENDIF
\ENDIF	
\UNTIL{task$\neq\perp$}
\end{algorithmic}
\end{algorithm}

Algorithm~\ref{alg:coordinate} lists the coordinateTask()
procedure.  It is called by the coordinating thread, and checks
whether execution of a task can start. If so, it signals that
execution can be started by setting $t$ to $r$ and storing the
currentTask into \emph{task} and starts execution.

\begin{algorithm}
\caption{The coordinateTask() procedure}\label{alg:coordinate}
\begin{algorithmic}[1]
\REPEAT
\STATE $RR \gets R$
\IF {$RR.r = RR.a$}
	\STATE \COMMENT {Enough threads have registered, attempt to fix the team}
	\STATE $RR'\gets RR$
	\STATE $RR'.t \gets RR.r$
	\IF {CAS($R$, $RR$, $RR'$)}
		\STATE \COMMENT {Team built! Start execution}
		\STATE $G \gets RR'.r-1$
		\STATE task $\gets$ $Q$.popBottom()
		\STATE $R.r\gets$ max($R.t$, $Q.\mbox{bottom}.r$)
	\ELSE
		\STATE backoff()
	\ENDIF
\ELSE
	\STATE pollPartners($I$, $R.r$)
\ENDIF
\UNTIL{task$\neq\perp$}
\end{algorithmic}
\end{algorithm}

The modified stealTasks() procedure is shown as
Algorithm~\ref{alg:newstealtask}.  It is called by a thread when its
queue has run empty. It tries to either find a coordinated task to
work on, or to steal tasks.  Partners are checked by flipping the bits
of the thread id from least to most significant, taking at most $\log
p$ iterations. This is done by bitwise exclusive or (denoted by
$\oplus$) with $2^{\ell}$ for $\ell=0$ to $\log p-1$.

\begin{algorithm}
\caption{Modified stealTasks()}\label{alg:newstealtask}
\begin{algorithmic}[1]
\STATE $\ell\gets 0$
\WHILE {$2^{\ell}<p$ }
	\STATE $x \gets \mbox{ThreadRef}[I \oplus 2^{\ell}]$ \COMMENT{Deterministic partner}
	\STATE $xc \gets x.c$ \COMMENT{The partner's coordinator}
	\STATE $xcR \gets xc.R$ \COMMENT{Copy  coordinators registration structure}
		
	\IF[Partner's coordinator requires this thread for execution of its task]{$xcR.r \geq 2^{\ell+1}$ }
		\STATE $RR \gets xcR$
		\STATE $RR.a \gets RR.a+1$ \COMMENT{Coordinator has
		acquired one more thread}
		\IF{CAS($x.R$, $xcR$, $RR$)}
			\STATE \COMMENT{Successful registration with new coordinator}
			\STATE $c \gets xc$ \COMMENT{Partner is new coordinator}
			\STATE coordinatorCounter $\gets RR.N$ 
			\RETURN
		\ENDIF
	\ELSE
	\STATE\COMMENT{Steal from partner instead}
		\STATE $T \gets\min(x.Q.\mbox{size}() / 2, 2^{\ell})$
		\IF {$Q.\mbox{popappend}(x.Q, T)>0$ } 
			\STATE \COMMENT {At least one task stolen}
			\RETURN
		\ENDIF
		\STATE\COMMENT{Nothing to steal, next partner}
		\STATE $\ell\gets\ell+1$
	\ENDIF
\ENDWHILE
\STATE \COMMENT {No success in stealing procedure}
\STATE backoff()
\end{algorithmic}
\end{algorithm}

The pollPartners($c$,$r$) procedure shown as Algorithm~\ref{alg:poll}
is the new polling method. It takes two parameters: A coordinator $c$
and the number of required threads $r$. It polls all partners required
for the execution of a task. If the partners are working on tasks that
are smaller than the current task to execute, we steal tasks from them
to help them complete faster and start looking for partners. If a
partner is also working on a large task (with large, we mean that it
requires many threads for execution), we have to make sure that
exactly one of them wins. In our case, we deterministically choose the
task with a smaller thread requirement $r$. If both tasks are of the
same size, the task of the thread with the smaller id wins. Although
it might be more intuitive to prefer larger tasks to smaller tasks, we
may not do this as threads are only guaranteed to find larger tasks as
soon as they run out of smaller tasks.

There might be better ways of breaking ties, e.g.\ based on size of
local queues, which might improve performance, and this is subject to
experimentation. It is easy to see that the chosen criterion is
correct.

\begin{algorithm}
\caption{pollPartners($c$,$r$)}\label{alg:poll}
\begin{algorithmic}[1]
\STATE $\ell\gets 0$
\WHILE {$2^{\ell}<r$}
	\STATE $x \gets \mbox{ThreadRef}[I \oplus 2^{\ell}]$ \COMMENT{Deterministic partner}
	\STATE $xc \gets x.c$ \COMMENT{Copy partners coordinator pointer}
	\STATE $xcR \gets xc.R$ \COMMENT{Copy partners registration structure}
	
	\IF {$xc.I\neq c.I$ }
		\IF {$xcR.r=r$}
			\IF {$xc.I<c.I$}
				\STATE \COMMENT {The partner's task wins. Switch to its task}
				\IF {task $\neq$ $\perp$}
					\STATE $Q$.pushBottom(task)
					\STATE task $\gets\perp$
				\ENDIF
				\STATE switchToCoordinator($xc$)
				\RETURN
			\ENDIF
			\STATE \COMMENT {We win, partner will eventually register for our task}
		\ELSIF {$xR.r<r$}
			\STATE \COMMENT {The partner's task wins.}
			\IF {$xR.r<2^{\ell+1}$}
				\STATE \COMMENT {The partner doesn't require help from this thread. We steal some tasks to make sure it's queues are empty sooner.}
				\IF {$x.Q.\mbox{size}() > 0$}
					\IF {task $\neq$ $\perp$}
						\STATE $Q$.pushBottom(task)
						\STATE task $\gets\perp$
					\ENDIF
					\STATE $T \gets\min(x.Q.\mbox{size}() / 2, 2^{\ell})$				
					\STATE $Q.\mbox{popappend}(x.Q, T)$
					\RETURN
				\ENDIF
			\ELSE
				\STATE \COMMENT {The partner requires help from this thread. We switch to its task}
				\IF {task $\neq$ $\perp$}
					\STATE $Q$.pushBottom(task)
					\STATE task $\gets\perp$
				\ENDIF
				\STATE switchToCoordinator($xc$)
				\RETURN
			\ENDIF
		\ENDIF
		\STATE \COMMENT {We win, partner will eventually register for our task}
	\ENDIF
	\STATE $\ell\gets\ell+1$
\ENDWHILE

\IF{$\neg$ $c$.taskIsReady($I$)} \STATE backoff()
\ENDIF
\end{algorithmic}
\end{algorithm}

The switchToCoordinator($xc$) procedure presented in
Algorithm~\ref{alg:switch} tries to set the coordinator reference $c$
to the coordinator given as parameter. It also performs deregistration
from the old coordinator if necessary, and checks whether the
coordinator still requires help from this thread, before registering
for it. The overlap() function used in this method checks whether the
given thread ids would both be in the same team for a task of the size
specified as the third parameter. This works similar to the
calculation of local thread id's described in Section~\ref{sec:prop}.

\begin{algorithm}
\caption{switchToCoordinator($xc$)}\label{alg:switch}
\begin{algorithmic}[1]
\LOOP
	\STATE $xcR \gets xc.R$
	\IF [Coordinator requires help from this thread]{overlap($xc.I$, $I$, $xcR.r$)}
		\IF {$c \neq \mbox{ThreadRef}[I]$}
			\STATE \COMMENT {First drop previous coordinator}
			\STATE $RR \gets c.R$
			\IF {overlap($c.I$, $I$, $RR.t$)}
				\STATE \COMMENT {We are in our current coordinators team and therefore can't drop out}
				\RETURN
			\ENDIF
			\STATE $RR' \gets RR$
			\STATE $RR'.a \gets RR'.a - 1$
			\IF {$\neg$ CAS($c.R$, $RR$, $RR'$)}
				\STATE backoff();
			\ELSE
				\RETURN
			\ENDIF
		\ELSE
			\STATE $RR \gets xcR$
			\STATE $RR.a \gets RR.a + 1$
			\IF {CAS($xc.R$, $xcR$, $RR$)}
				\STATE \COMMENT {We have successfully registered for coordinator}
				\IF{$R.r > 1$}
					\STATE \COMMENT{If this thread was coordinating a task, we have to stop coordinating}
					\STATE $RR' \gets R$
					\STATE $RR'.r \gets 1$
					\STATE $RR'.t \gets 1$
					\STATE $RR'.a \gets 1$
					\STATE $RR'.N \gets RR'.N + 1$
					\STATE $R \gets RR'$
				\ENDIF
				\STATE $c \gets xc$
				\STATE $cN \gets RR.N$
			\ELSE
				\STATE backoff()
			\ENDIF
		\ENDIF
	\ELSE
		\RETURN
	\ENDIF
\ENDLOOP

\end{algorithmic}
\end{algorithm}

\subsection{Basic properties}
\label{sec:prop}

Teams are always built out of consecutive threads, as the threads
allowed to join a team of a certain size at a certain coordinator are
static and deterministic as determined by the bit-flipping in the
stealTasks() and pollPartners($c$,$r$) procedures.  If we switch to a
different task, where some registered threads might not be required,
we reset the acquired counter $a$ to enforce this property. Due to
this bit-flipping, teams always consist of the thread id's
$kr,kr+1,\ldots,i,\ldots (k+1)r-1$ for some $k$ in the range $0 \leq k
< p/r$.

Teams stay together as long as the coordinator's next task is
the same size as the team. If the next task is smaller, the team is
deterministically shrunk to the required size. This is done by the
coordinator by updating $t$. Each thread can deterministically
calculate, whether it is still on the team. If the next task is
larger, the coordinator breaks up the team as soon as execution of the
previous task has finished. This is done, by setting $t=1$. The team
for the larger task then has to be rebuilt from scratch.

Stealing follows a deterministic pattern in our scheduler. We contact
$\log p$ partner threads, before backing off. This was necessary in
order for the teams to build properly, and may furthermore be advantageous
to ensure memory-locality.  If threads are initially numbered such that
threads within each  memory-hierarchy level are consecutive,
bit-flipping will ensure that teams are formed by threads that are
close in the memory hierarchy.  Such locality optimizations by
deterministic stealing have often been considered, see for instance
the BubbleSched framework~\cite{Thibault07}. 

An important property of work-stealing with double ended queues is
that tasks are executed in depth-first order. With deterministic
team-building it can happen that larger ($r>2$) tasks are
stolen back and forth until they are finally executed if there are smaller tasks in-between. This may also
lead to two tasks of the same size switching order inside a queue,
therefore violating the depth-first order. We address this issue in
Refinement 1. There we also show correctness of the algorithm. 

Another property of work-stealing is that as long as a thread can
execute tasks it does not have to communicate with other threads. We
can expand this property to teams of arbitrary sizes, with the
restriction that this only holds as long as the next task requires the
same amount of threads as the previous one. Of course, communication
cannot be completely omitted with tasks requiring more than one
thread, as threads in a team have to poll the coordinator for the next
task, and have to notify it when execution starts, but this overhead
is small.  If task sizes in a single queue vary, communication is
needed every time a larger task follows a smaller task. This issue
also becomes less problematic with Refinement 1, although it is not
completely resolved if we allow arbitrary task sizes (Refinement 2).

We finally explain how a consecutive numbering of the threads in a
team starting from 0 is achieved.  As soon as a thread knows the size
of the team, it can use its global thread id to calculate the
boundaries of the team, and therefore its local id. This is done by
first retrieving the position of the most significant bit in
$t$. Retrieving the most significant bit can either be done in $\log
b$ operations where $b$ is the number of bits in the integer, but most
modern processors support this operation in hardware. The leftmost
thread id in the team is calculated by setting all bits in the
coordinator id that are below the most significant bit of $t$ to
$0$. For the rightmost thread id we have to set all those bits to
$1$. The local id's for the execution of a task can simply be
calculated by subtracting the leftmost thread id from the id of the
actual thread.

We estimate the extra overheads in deterministic team-building as
follows: an extra CAS used in Algorithms~\ref{alg:coordinate}
and~\ref{alg:newstealtask}. If all tasks require $r=1$ the algorithm
coincides with a deterministic work-stealing scheduler, where $\log p$
fixed partners are tried before the backoff(). The additional CAS that
do not appear in classic work-stealing are never executed in this
case. Actually, as now written in Algorithm~\ref{alg:coordinate} a CAS
in coordinateTask() would be executed as some code was omitted for
readability reasons. In the actual implementation, the CAS is only
executed if the new team size differs from the old one, and this is
never the case for one-thread tasks.

\subsection{Refinement 1: Multiple work queues}

In the basic variant, two tasks may switch order in a queue, if there
are some smaller tasks between them (by being stolen back and forth,
as explained above). This means that in the worst case tasks are not
executed in a depth-first order any more by a single thread. Also,
larger tasks might switch often between two queues, until most smaller
tasks are processed. We can resolve this problem by using $\log p$
local queues instead of one.

Each queue stores tasks of a certain size, in particular queue $Q_i$
keeps tasks requiring $r=2^i$ threads. A thread always executes the
smallest tasks first, and moves to queues with larger tasks as soon as
all queues with smaller tasks are empty. When a team of threads works
on a queue, it continues working on this queue, even if queues
containing smaller tasks get filled again. Only after the queue is
empty, the team is resized to work on a queue containing smaller
tasks.

We can now forbid threads to steal tasks, where both threads would
be in the same team. This reduces the required communication.

This refinement improves some of the properties of the algorithm. The
main improvement is, as described before, that tasks of a certain size
are now executed in a depth-first order. Reordering of such tasks is
impossible. Also, on one thread, tasks requiring less threads are
executed before larger tasks. The only exception occurs for small
tasks that are created after a larger team has been formed by this
method. They have to wait until the team is resized. Due to the
clustering of the execution of same-size tasks, we reduce the required
coordination due to varying task sizes.

This refinement is necessary for the following two refinements.

\subsection{Correctness}

For our correctness argument, we assume we have $\log p$ queues per
thread (as per Refinement 1), and that the number of threads required
per task, as well as the total number of threads, are powers of two.

\begin{lem} 
Assume that the computation is finite. A thread $i$ has a task
requiring $r\geq 1$ threads. This task will eventually be executed.
\end{lem}

\begin{proof}
For $r=1$ the case is clear. A task requiring a single thread will in
general be executed before tasks using more threads. No coordination
is required before execution, so the task will eventually be executed,
similar to classical work-stealing.

Similar findings apply to $r>1$. If we assume that all tasks in the
computation require $r$ threads, all threads will coordinate to join
teams of $r$ threads with $i\in [kr,(k+1)r-1]$. Assume, thread $i$ is
the coordinator, then the task will eventually be executed. Otherwise,
the team will dissolve to search for another coordinator as soon as
the current coordinator's queue runs empty, and eventually thread $i$
will become a coordinator.

If we relax the restriction that all tasks require $r$ threads, the
given task will be executed at latest after we run out of tasks
requiring less than $r$ threads. Tasks with thread requirements larger
than $r$ cannot block execution of the given task, as tasks requiring
less threads are always prioritized.
\end{proof}

\begin{lem}
Assume, we have two tasks $x$ and $y$ in the same queue with $n\geq 0$
tasks in-between them. When using $\log p$ queues per thread, $x$ and
$y$ cannot be reordered inside a single queue.
\end{lem}

\begin{proof}
Let's assume that $x$ is nearer to the top of the queue than
$y$. Therefore, $x$ would be stolen first. Assume, both get stolen by
the same thread, then the order of both tasks in the target queue
would stay the same, even if stolen at different times. The only case
when $x$ and $y$ could switch order would be if a thread has $y$ in
its queue, and then steals $x$. A task can only be stolen in two
cases: If all queues of the stealing thread are empty, or during
coordination. If all queues of the stealing thread are empty, they
cannot contain $y$. During coordination, only tasks are stolen that
require less threads than the task to coordinate. As coordination is
always done for the task that requires the least amount of threads, we
require that $r_{x}<r_{y}$, which contradicts our assumption that both
$x$ and $y$ are in the same queue.
\end{proof}

\begin{lem}
All conflicts are resolved deterministically.
\end{lem}

\begin{proof}
Assume that thread $x$ and thread $y$ both try to coordinate a task
with $y \in [kr_{x},(k+1)r_{x}-1]$ and $x \in
[kr_{y},(k+1)r_{y}-1]$. Assume $r_{x}=r_{y}$, and $x<y$ then $x$ will
be chosen deterministically. All threads with $y$ as coordinator will
switch to $x$ as soon as they encounter a thread with $x$ as
coordinator during coordination. Assume $r_{x}<r_{y}$, then again $x$
will be deterministically chosen. Assume that thread $x$ wins, but
thread $y$ steals another task during coordination of $y$ before
encountering threads coordinated by $x$. As we only allow to steal
tasks during coordination that require less threads than the
coordinated task, the new thread requirement $r'_y$ must be less than
the old requirement. Assume that now $r_{x}>r'_{y}$, then if $x \in
[kr'_{y},(k+1)r'_{y}-1]$ still holds, thread $x$ will switch to thread
$y$ as coordinator. Otherwise, thread $y$ is independent of thread $x$
and the conflict therefore resolved. As each stolen task has to be
smaller than the previous one, the conflict will be resolved either
way sooner or later.
\end{proof}

\begin{lem}
Each task is only executed once by each of the threads in a team.
\end{lem}

\begin{proof}
A task is always managed by only one thread (the coordinator) and
cannot occur in two queues at the same time. The start of task
execution is managed by the coordinator, and the reference to an
executed task removed before the coordinator starts coordinating
again. Coordinated threads in a team have to remember the last
executed task to make sure they do not execute it again, until they
either drop out of the team, or the coordinator starts coordinating a
new task.
\end{proof}

\subsection{Refinement 2: Arbitrary thread requirements}

We now show how to cope with the case where each new task can require
an arbitrary number of threads, $r\leq p$, and not only requirements
that are powers of two.

The easiest way to do this would be to just allocate a team with a
size equal to the next-highest power of two, and to let some threads
sit idle during execution. This, of course, is far from ideal, and it
would be preferable if the threads that would otherwise be idle worked
on smaller tasks. Nonetheless, we cannot completely ignore those
threads, as they might be the first partners, some thread that is
required for the team visits.

We propose that during coordination, such threads that will not
actually work on a task  silently register at the
coordinator. Registering silently means, that the thread's
coordination pointer is set to the coordinator, but it does not
increment the registration counter. As soon as execution of the task
starts, the thread may start working on another task.

We note that it is still necessary to help those tasks empty their
queues, even if they might not always interfere in coordination and
might later run out of work. Sometimes, some of those threads
might be coordinating another task that requires a team that does not
intersect with the team of our task. We do not need to steal from those
threads as they won't interfere with our task.

Although it is possible to support arbitrary task sizes, we can only
provide weak guarantees concerning the utilization of the hardware
threads. In the worst case, nearly half of the threads may sit
idle. This would happen if we have tasks with $r=2^{k}+1$ to execute,
and all smaller tasks on silently registered threads would have been
executed before forming the team. Therefore the programmer should
preferably use tasks that are aligned to a power of two.

Another problem is that teams might dissolve before a queue has been
processed, because of a larger task following a smaller
task. Therefore the programmer should try not to have varying task
sizes in single queues. For some applications it might be feasible to
provide additional queues for certain sizes that are often used, but
this approach cannot be generalized, and providing $p$ queues is not
feasible anyway with increasing number of cores.

\subsection{Refinement 3: Arbitrary number of hardware threads}

We finally extend to the general situation where an arbitrary (finite,
fixed) number of threads is given from the outset, and
each newly spawned task can require an arbitrary number of threads.

In the standard algorithm, we assume that each level $\ell$ at which a
thread has one partner, contains exactly $2^{\ell}$ threads. We relax
this constraint by allowing a level to contain $n_{\ell}$ tasks, where
$n_{\ell-1} < n_{\ell} \leq 2n_{\ell-1}$ and $n_{0}=1$. This
information has to be statically precomputed at startup time, and has to
be accessible to all threads. This relaxation has two implications:
First, some threads will not have a partner at certain levels
$\ell$. Second, some threads won't have access to the full number of
threads for a team $n_{\ell}$ on level $\ell$.

As the information about a thread's partner cannot be conveniently
generated on the fly any more, each thread has to precompute and
store an array $P$ of its $\log p$ partners. If, for a thread, the
partner at level $\ell$ is missing, we store $P[{\ell}]=\perp$. Also,
each thread has to precompute and store the actual team-sizes $n'$
available at level $\ell$, where $n'_{\ell-1} \leq n'_{\ell} \leq
n_{\ell}$. Each thread has $\log p$ task-queues, where the queue for
level $\ell$ stores tasks in the range $n'_{\ell-1} < x \leq
n'_{\ell}$. Some threads might have queues that are never used, and
therefore do not have to be reserved in memory.

The actual execution proceeds similarly to the standard execution,
only that instead of relying on bit-flipping, we have to rely on
precomputed information about partners and team-sizes. Also, as
partners at some levels might not be available, we should be able to
handle that. Last, but not least, we have to be aware that stolen
tasks will not necessarily be stored in a queue at the same level as
in the originating thread. This might create balancing issues, where
tasks are not stolen from a partner, as they are in a queue at the
same level as the level of the partner, which is not allowed by the
algorithm, even though two or more of those tasks could be executed in
parallel by all threads at this level. This case may actually only
occur, if a thread has queues that would never be used as described in
the previous paragraph. If we use those queues for storing the tasks
in question, we can resolve the issue. 

The properties of the algorithm should more or less stay the same with
this refinement, only that the ideal task-size is no longer a power of
two and may vary depending on the actual thread the task is executed
on. An advantage of this approach is that it can sometimes provide a
more suitable representation of a homogeneous multi-core and therefore
provide good locality in many cases. For example, if we take a
dual-socket system with two three-core processors, we would structure
the threads with $n_{0}=2$, $n_{1}=3$ and $n_{2}=6$. This guarantees,
that a 3-processor task is executed on one core, which reduces memory
access times.

\subsection{Refinement 4: Randomizing stealing and team-building}

The deterministic work-stealing scheme can be augmented with
randomization which can theoretically guard against (easily
constructible) degenerate cases where threads may sit idle waiting for
any partner thread to finish working on its current task and starting
to steal from threads that this thread may not reach. When choosing a
partner for stealing and/or team-building at a certain level $\ell$,
in addition to applying a bitwise exclusive or with $2^{\ell}$ to the
id of the stealing thread, we also randomize all bits below the
$\ell$th bit. This can be implemented by performing the exclusive or
with a random integer $2^{\ell}\leq i<2^{\ell+1}-1$ instead of the
fixed bit $2^{\ell}$.  Thus, instead of always deterministically
choosing the same thread at level $\ell$, a random thread is chosen
out of a set of $2^{\ell}$ threads. Using this strategy, each thread
may steal from any thread over a total of $\log p$ steal attempts, but
the required hierarchy between the threads is preserved.

\section{Implementation}
\label{sec:implementation}

We have implemented a prototype of the work-stealing scheduler with
deterministic team-building as described above in C++ using Pthreads
to start the $p$ hardware threads.  The atomic operations used in the
implementation are \emph{compare-and-swap} and
\emph{fetch-and-decrement}, which are all available as atomic builtins
in \texttt{gcc}. The \emph{compare-and-swap} primitive is required for
modifications on the registration structure, and for accesses to the
work-stealing deque. \emph{fetch-and-decrement} is used for counting
down started tasks. For retrieving the most significant bit of an
integer, we use the \emph{bsrl} assembly instruction available on
Intel architectures, as this operation is not provided as a library
call under the Linux operating system. Under BSD, the
\emph{fls} library function can be used instead. Retrieving the most
significant bit is necessary for calculating the boundaries of a team
as explained in Section~\ref{sec:prop}, and for choosing in which
queue to store a task.

Furthermore, the following design decisions have been made for the
implementation:

\begin{itemize}
\item Tasks are implemented as objects derived from a base task class,
quite similar to TBB~\cite{KukanovVoss07}.
\item For simplicity, we only provide one linear stack per thread in
our implementation. A cactus-stack as used in
Cilk~\cite{BlumofeJoergKuszmaulLeisersonRandallZhou96} might be more
efficient.
\item When stealing tasks the last stolen task is not put on the
stack but instead returned immediately from the stealtasks()
function. This is necessary to prevent situations, where a task is
stolen back and forth with no thread being able to execute it.
\item The scheduler terminates as soon as all threads have registered
as idle. They can register as idle if their stack and all queues are
empty and stealing has failed multiple times. Registration is
canceled before a thread starts to steal again.
\item
We have noticed that we can achieve better scheduling in many
cases, if we steal the largest allowed tasks. This comes from the fact
that a thread only steals from a thread at a certain level, if all
partner threads at lower levels had empty queues. Therefore, the
chances are high that the stealing thread will be able to build up a
team soon.
\end{itemize}

Some of tunable parameters of the implementation are given
below. Performance of the implementation might be improved by choosing
the right values, and the optimal values might differ depending on the
hardware the scheduler is run on.

\begin{itemize}
\item Backoff intervals - For our backoff function, we used
exponential backoff, starting at 1 microsecond, and going up to 10
milliseconds.
\item Number of tasks to steal - We decided to steal $2^{\ell}$ tasks
from a partner, where $\ell$ is the position of the bit to flip to get
the partner's id. This comes from the assumption that, if we reached
the $\ell^{th}$ partner during stealing, it is likely that all threads
in the $2^{\ell}$ block around the current task are running out of
tasks. Therefore it makes sense to steal enough tasks for all of them.
\end{itemize}

\section{An example with experimental results}
\label{sec:quicksort}

\begin{table*}[t!]
\begin{center}
\begin{tabular}{|lr|rr|rrr|rrr||rr|}
\hline
Type & Size & Seq/STL & SeqQS & Fork & SU & Randfork & Cilk & SU & Cilk\_sample & MMPar & SU \\
\hline
 & 10000000 & 0.940 & 1.022 & 0.243 & 3.9 & 1.027 & 0.163 & 5.8 & 0.185 & 0.201 & 4.7 \\
 & 100000000 & 10.492 & 11.421 & 2.244 & 4.7 & 9.085 & 1.828 & 5.7 & 1.953 & 1.669 & 6.3 \\
Random & 1000000000 & 112.110 & 122.450 & 20.964 & 5.3 & 31.643 & 18.903 & 5.9 & 20.534 & 18.130 & 6.2 \\
 & 8388607 & 0.781 & 0.848 & 0.229 & 3.4 & 0.864 & 0.154 & 5.1 & 0.158 & 0.182 & 4.3 \\
 & 33554431 & 3.320 & 3.639 & 0.778 & 4.3 & 3.357 & 0.587 & 5.7 & 0.681 & 0.603 & 5.5 \\
 & 134217727 & 14.335 & 15.638 & 2.924 & 4.9 & 9.422 & 2.112 & 6.8 & 2.556 & 2.236 & 6.4 \\
\hline
 & 10000000 & 0.937 & 1.017 & 0.245 & 3.8 & 1.189 & 0.154 & 6.1 & 0.184 & 0.199 & 4.7 \\
 & 100000000 & 9.971 & 10.883 & 2.310 & 4.3 & 7.713 & 2.025 & 4.9 & 2.280 & 1.646 & 6.1 \\
Gauss & 1000000000 & 101.042 & 110.295 & 20.151 & 5.0 & 34.062 & 18.385 & 5.5 & 24.096 & 16.580 & 6.1 \\
 & 8388607 & 0.785 & 0.847 & 0.205 & 3.8 & 0.794 & 0.139 & 5.7 & 0.156 & 0.177 & 4.4 \\
 & 33554431 & 3.328 & 3.609 & 0.727 & 4.6 & 3.403 & 0.604 & 5.5 & 0.649 & 0.594 & 5.6 \\
 & 134217727 & 13.613 & 14.859 & 2.881 & 4.7 & 10.175 & 2.171 & 6.3 & 2.625 & 2.103 & 6.5 \\
\hline
 & 10000000 & 0.873 & 0.962 & 0.255 & 3.4 & 1.121 & 0.117 & 7.5 & 0.141 & 0.204 & 4.3 \\
 & 100000000 & 10.493 & 11.691 & 1.921 & 5.5 & 9.489 & 1.366 & 7.7 & 1.687 & 1.610 & 6.5 \\
Buckets & 1000000000 & 108.909 & 121.008 & 21.008 & 5.2 & 31.683 & 14.691 & 7.4 & 18.208 & 17.451 & 6.2 \\
 & 8388607 & 0.721 & 0.785 & 0.174 & 4.2 & 0.774 & 0.088 & 8.2 & 0.113 & 0.174 & 4.2 \\
 & 33554431 & 3.205 & 3.535 & 0.615 & 5.2 & 3.262 & 0.415 & 7.7 & 0.502 & 0.561 & 5.7 \\
 & 134217727 & 13.573 & 14.971 & 2.344 & 5.8 & 8.555 & 1.465 & 9.3 & 1.945 & 2.129 & 6.4 \\
\hline
 & 10000000 & 0.869 & 0.977 & 0.219 & 4.0 & 1.041 & 0.148 & 5.9 & 0.170 & 0.189 & 4.6 \\
 & 100000000 & 9.845 & 10.837 & 1.814 & 5.4 & 6.621 & 1.154 & 8.5 & 1.480 & 1.593 & 6.2 \\
Staggered & 1000000000 & 102.498 & 112.668 & 17.593 & 5.8 & 24.018 & 13.869 & 7.4 & 18.701 & 16.096 & 6.4 \\
 & 8388607 & 0.731 & 0.819 & 0.173 & 4.2 & 0.647 & 0.173 & 4.2 & 0.180 & 0.174 & 4.2 \\
 & 33554431 & 3.120 & 3.591 & 0.701 & 4.5 & 2.612 & 0.387 & 8.1 & 0.485 & 0.613 & 5.1 \\
 & 134217727 & 13.365 & 14.816 & 2.356 & 5.7 & 8.746 & 1.720 & 7.8 & 2.050 & 2.174 & 6.1 \\

\hline
\end{tabular}
\end{center}
\caption{Quicksort on the 8-core Intel Nehalem system. 
Average running times over 10 repetitions in seconds. 
Speedup is calculated relative to the (best) sequential STL implementation.}
\label{tab:sort_cora_avg}
\end{table*}

\begin{table*}[t!]
\begin{center}
\begin{tabular}{|lr|rr|rrr|rrr||rr|}
\hline
Type & Size & Seq/STL & SeqQS & Fork & SU & Randfork & Cilk & SU & Cilk\_sample & MMPar & SU \\
\hline
 & 10000000 & 0.939 & 1.017 & 0.232 & 4.0 & 0.505 & 0.162 & 5.8 & 0.183 & 0.194 & 4.8 \\
 & 100000000 & 10.483 & 11.404 & 2.168 & 4.8 & 4.813 & 1.812 & 5.8 & 1.911 & 1.641 & 6.4 \\
Random & 1000000000 & 111.442 & 121.697 & 20.770 & 5.4 & 23.703 & 18.665 & 6.0 & 20.441 & 16.973 & 6.6 \\
 & 8388607 & 0.767 & 0.834 & 0.215 & 3.6 & 0.696 & 0.152 & 5.0 & 0.158 & 0.173 & 4.4 \\
 & 33554431 & 3.317 & 3.632 & 0.765 & 4.3 & 1.316 & 0.585 & 5.7 & 0.646 & 0.577 & 5.7 \\
 & 134217727 & 14.240 & 15.535 & 2.853 & 5.0 & 3.524 & 2.101 & 6.8 & 2.550 & 2.213 & 6.4 \\
\hline
 & 10000000 & 0.926 & 1.006 & 0.238 & 3.9 & 1.086 & 0.154 & 6.0 & 0.183 & 0.187 & 4.9 \\
 & 100000000 & 9.961 & 10.864 & 2.250 & 4.4 & 4.002 & 2.014 & 4.9 & 2.262 & 1.573 & 6.3 \\
Gauss & 1000000000 & 100.551 & 109.778 & 19.900 & 5.1 & 23.386 & 18.273 & 5.5 & 24.036 & 15.567 & 6.5 \\
 & 8388607 & 0.765 & 0.826 & 0.193 & 4.0 & 0.304 & 0.138 & 5.5 & 0.155 & 0.168 & 4.6 \\
 & 33554431 & 3.275 & 3.555 & 0.704 & 4.7 & 1.377 & 0.599 & 5.5 & 0.646 & 0.568 & 5.8 \\
 & 134217727 & 13.607 & 14.830 & 2.865 & 4.7 & 3.751 & 2.163 & 6.3 & 2.617 & 2.091 & 6.5 \\
\hline
 & 10000000 & 0.864 & 0.950 & 0.234 & 3.7 & 0.980 & 0.116 & 7.5 & 0.140 & 0.191 & 4.5 \\
 & 100000000 & 10.050 & 11.190 & 1.893 & 5.3 & 5.331 & 1.357 & 7.4 & 1.679 & 1.583 & 6.3 \\
Buckets & 1000000000 & 104.524 & 116.298 & 20.745 & 5.0 & 24.175 & 14.535 & 7.2 & 18.146 & 16.843 & 6.2 \\
 & 8388607 & 0.711 & 0.774 & 0.169 & 4.2 & 0.404 & 0.087 & 8.2 & 0.112 & 0.160 & 4.5 \\
 & 33554431 & 3.092 & 3.408 & 0.593 & 5.2 & 1.931 & 0.410 & 7.5 & 0.499 & 0.546 & 5.7 \\
 & 134217727 & 13.363 & 14.724 & 2.315 & 5.8 & 2.951 & 1.458 & 9.2 & 1.942 & 2.097 & 6.4 \\
\hline
 & 10000000 & 0.865 & 0.971 & 0.208 & 4.2 & 0.610 & 0.144 & 6.0 & 0.165 & 0.174 & 5.0 \\
 & 100000000 & 9.837 & 10.816 & 1.785 & 5.5 & 2.940 & 1.146 & 8.6 & 1.475 & 1.569 & 6.3 \\
Staggered & 1000000000 & 101.711 & 111.983 & 17.368 & 5.9 & 19.766 & 13.595 & 7.5 & 18.567 & 15.823 & 6.4 \\
 & 8388607 & 0.722 & 0.807 & 0.167 & 4.3 & 0.272 & 0.171 & 4.2 & 0.177 & 0.161 & 4.5 \\
 & 33554431 & 3.119 & 3.581 & 0.662 & 4.7 & 1.506 & 0.379 & 8.2 & 0.480 & 0.566 & 5.5 \\
 & 134217727 & 13.357 & 14.796 & 2.325 & 5.7 & 3.017 & 1.698 & 7.9 & 2.027 & 2.075 & 6.4 \\
\hline
\end{tabular}
\end{center}
\caption{Quicksort on the 8-core Intel Nehalem system. 
Best (minimum) running time over 10 runs in seconds.
Speedup is calculated relative to the (best) sequential STL implementation.}
\label{tab:sort_cora_min}
\end{table*}

\begin{table*}[t!]
\begin{center}
\begin{tabular}{|lr|rr|rrr||rr|}
\hline
Type & Size & Seq/STL & SeqQS & Fork & SU & Randfork & MMPar & SU \\
\hline
& 10000000 & 1.305 & 1.268 & 0.581 & 2.2 & 1.254 & 0.782 & 1.7 \\
& 100000000 & 14.890 & 14.575 & 3.710 & 4.0 & 11.836 & 3.164 & 4.7 \\
Random & 8388607 & 1.106 & 1.053 & 0.457 & 2.4 & 1.116 & 0.502 & 2.2 \\
& 33554431 & 4.751 & 4.653 & 1.291 & 3.7 & 4.756 & 1.252 & 3.8 \\
& 134217727 & 20.948 & 19.972 & 4.466 & 4.7 & 18.034 & 4.427 & 4.7 \\
\hline
 & 10000000 & 1.283 & 1.260 & 0.503 & 2.6 & 1.341 & 0.647 & 2.0 \\
 & 100000000 & 14.356 & 13.994 & 3.540 & 4.1 & 12.216 & 2.902 & 4.9 \\
Gauss & 8388607 & 1.056 & 1.058 & 0.478 & 2.2 & 1.055 & 0.517 & 2.0 \\
 & 33554431 & 4.734 & 4.503 & 1.342 & 3.5 & 4.381 & 1.799 & 2.6 \\
 & 134217727 & 19.997 & 19.244 & 5.160 & 3.9 & 16.887 & 4.718 & 4.2 \\
\hline
 & 10000000 & 1.291 & 1.212 & 0.488 & 2.6 & 1.272 & 0.821 & 1.6 \\
 & 100000000 & 14.734 & 14.035 & 3.412 & 4.3 & 13.230 & 3.355 & 4.4 \\
Buckets & 8388607 & 1.071 & 0.967 & 0.403 & 2.7 & 1.114 & 0.583 & 1.8 \\
 & 33554431 & 4.670 & 4.515 & 1.266 & 3.7 & 4.265 & 1.497 & 3.1 \\
 & 134217727 & 20.666 & 19.031 & 4.351 & 4.7 & 15.014 & 4.046 & 5.1 \\
\hline
 & 10000000 & 1.187 & 1.306 & 0.631 & 1.9 & 1.350 & 0.828 & 1.4 \\
 & 100000000 & 13.897 & 14.800 & 4.341 & 3.2 & 11.857 & 3.590 & 3.9 \\
Staggered & 8388607 & 1.064 & 1.058 & 0.440 & 2.4 & 1.213 & 0.671 & 1.6 \\
 & 33554431 & 4.597 & 4.631 & 1.216 & 3.8 & 4.775 & 1.611 & 2.9 \\
 & 134217727 & 19.133 & 19.660 & 4.844 & 4.0 & 15.354 & 4.399 & 4.3 \\
\hline
\end{tabular}
\end{center}
\caption{Quicksort on the 16-core AMD Opteron system. 
Average running times over 10 repetitions in seconds.
Speedup is calculated relative to the (best) sequential STL implementation.}
\label{tab:sort_daisy_avg}
\end{table*}

\begin{table*}[t!]
\begin{center}
\begin{tabular}{|lr|rr|rrr||rr|}
\hline
Type & Size & Seq/STL & SeqQS & Fork & SU & Randfork & MMPar & SU \\
\hline
 & 10000000 & 1.305 & 1.267 & 0.536 & 2.4 & 0.929 & 0.676 & 1.9 \\
 & 100000000 & 14.884 & 14.574 & 3.614 & 4.1 & 7.481 & 2.896 & 5.1 \\
Random & 8388607 & 1.106 & 1.052 & 0.423 & 2.6 & 0.608 & 0.436 & 2.5 \\
 & 33554431 & 4.751 & 4.653 & 1.233 & 3.9 & 4.254 & 1.069 & 4.4 \\
 & 134217727 & 20.947 & 19.971 & 4.302 & 4.9 & 10.399 & 4.119 & 5.1 \\
\hline
 & 10000000 & 1.282 & 1.260 & 0.466 & 2.8 & 1.092 & 0.568 & 2.3 \\
 & 100000000 & 14.349 & 13.993 & 3.429 & 4.2 & 9.069 & 2.699 & 5.3 \\
Gauss & 8388607 & 1.056 & 1.058 & 0.407 & 2.6 & 0.621 & 0.406 & 2.6 \\
 & 33554431 & 4.733 & 4.503 & 1.294 & 3.7 & 2.840 & 1.368 & 3.5 \\
 & 134217727 & 19.989 & 19.233 & 4.862 & 4.1 & 9.264 & 4.279 & 4.7 \\
\hline
 & 10000000 & 1.290 & 1.211 & 0.344 & 3.7 & 0.734 & 0.734 & 1.8 \\
 & 100000000 & 14.732 & 14.026 & 3.153 & 4.7 & 8.399 & 3.096 & 4.8 \\
Buckets & 8388607 & 1.071 & 0.967 & 0.355 & 3.0 & 1.102 & 0.531 & 2.0 \\
 & 33554431 & 4.669 & 4.515 & 1.138 & 4.1 & 2.498 & 1.294 & 3.6 \\
 & 134217727 & 20.655 & 19.030 & 3.933 & 5.3 & 9.265 & 3.835 & 5.4 \\
\hline
 & 10000000 & 1.187 & 1.306 & 0.609 & 2.0 & 0.762 & 0.732 & 1.6 \\
 & 100000000 & 13.889 & 14.793 & 3.820 & 3.6 & 6.676 & 3.117 & 4.5 \\
Staggered & 8388607 & 1.063 & 1.058 & 0.399 & 2.7 & 1.182 & 0.575 & 1.8 \\
 & 33554431 & 4.596 & 4.631 & 1.121 & 4.1 & 3.654 & 1.405 & 3.3 \\
 & 134217727 & 19.129 & 19.659 & 4.613 & 4.1 & 10.233 & 3.955 & 4.8 \\
\hline
\end{tabular}
\end{center}
\caption{Quicksort on the 16-core AMD Opteron system. 
Best (minimum) running time of 10 runs in seconds.
Speedup is calculated relative to the (best) sequential STL implementation.}
\label{tab:sort_daisy_min}
\end{table*}

\begin{table*}[t!]
\begin{center}
\begin{tabular}{|lr|rr|rrr|rrr||rr|}
\hline
Type & Size & Seq/STL & SeqQS & Fork & SU & Randfork & Cilk & SU & Cilk\_sample & MMPar & SU \\
\hline
 & 10000000 & 1.479 & 1.620 & 0.388 & 3.8 & 1.818 & 0.207 & 7.1 & 0.206 & 0.246 & 6.0 \\
 & 100000000 & 13.319 & 13.742 & 2.891 & 4.6 & 13.607 & 2.421 & 5.5 & 2.312 & 1.372 & 9.7 \\
Random & 1000000000 & 107.080 & 117.963 & 20.287 & 5.3 & 50.679 & 24.018 & 4.5 & 23.838 & 14.200 & 7.5 \\
 & 8388607 & 1.447 & 1.580 & 0.774 & 1.9 & 1.772 & 0.194 & 7.5 & 0.188 & 0.410 & 3.5 \\
 & 33554431 & 4.863 & 5.265 & 0.903 & 5.4 & 5.690 & 0.657 & 7.4 & 0.641 & 0.587 & 8.3 \\
 & 134217727 & 15.888 & 16.617 & 3.103 & 5.1 & 12.115 & 2.525 & 6.3 & 2.521 & 1.835 & 8.7 \\
\hline
 & 10000000 & 1.252 & 1.354 & 0.275 & 4.6 & 1.621 & 0.175 & 7.1 & 0.175 & 0.174 & 7.2 \\
 & 100000000 & 11.923 & 12.971 & 2.516 & 4.7 & 14.972 & 2.433 & 4.9 & 2.484 & 1.456 & 8.2 \\
Gauss & 1000000000 & 119.464 & 130.255 & 22.288 & 5.4 & 106.658 & 24.641 & 4.8 & 24.789 & 17.397 & 6.9 \\
 & 8388607 & 1.029 & 1.112 & 0.247 & 4.2 & 1.353 & 0.174 & 5.9 & 0.174 & 0.169 & 6.1 \\
 & 33554431 & 4.408 & 4.236 & 0.870 & 5.1 & 5.492 & 0.734 & 6.0 & 0.712 & 0.543 & 8.1 \\
 & 134217727 & 15.888 & 17.263 & 2.771 & 5.7 & 19.774 & 2.479 & 6.4 & 2.530 & 1.763 & 9.0 \\
\hline
 & 10000000 & 1.131 & 1.233 & 0.241 & 4.7 & 1.517 & 0.134 & 8.4 & 0.142 & 0.181 & 6.2 \\
 & 100000000 & 12.373 & 12.801 & 1.818 & 6.8 & 15.136 & 1.080 & 11.5 & 1.094 & 1.416 & 8.7 \\
Buckets & 1000000000 & 122.822 & 135.833 & 19.214 & 6.4 & 121.967 & 16.566 & 7.4 & 17.721 & 15.072 & 8.1 \\
 & 8388607 & 0.969 & 1.057 & 0.186 & 5.2 & 1.244 & 0.077 & 12.5 & 0.083 & 0.169 & 5.7 \\
 & 33554431 & 4.111 & 4.505 & 0.662 & 6.2 & 4.774 & 0.518 & 7.9 & 0.560 & 0.516 & 8.0 \\
 & 134217727 & 16.484 & 17.154 & 2.038 & 8.1 & 17.203 & 1.844 & 8.9 & 2.005 & 1.787 & 9.2 \\
\hline
 & 10000000 & 1.151 & 1.301 & 0.279 & 4.1 & 1.509 & 0.396 & 2.9 & 0.431 & 0.182 & 6.3 \\
 & 100000000 & 12.181 & 12.498 & 1.618 & 7.5 & 14.449 & 4.109 & 3.0 & 4.295 & 1.470 & 8.3 \\
Staggered & 1000000000 & 116.734 & 131.596 & 20.067 & 5.8 & 100.270 & 78.455 & 1.5 & 83.268 & 23.365 & 5.0 \\
 & 8388607 & 0.971 & 1.140 & 0.339 & 2.9 & 1.330 & 0.371 & 2.6 & 0.397 & 0.191 & 5.1 \\
 & 33554431 & 4.116 & 4.527 & 0.623 & 6.6 & 5.042 & 1.014 & 4.1 & 1.111 & 0.486 & 8.5 \\
 & 134217727 & 16.281 & 16.941 & 2.299 & 7.1 & 17.563 & 2.146 & 7.6 & 2.243 & 1.904 & 8.5 \\
\hline
\end{tabular}
\end{center}
\caption{Quicksort on the 32-core Intel Nehalem EX system. 
Average running timesover 10 repetitions in seconds.
Speedup is calculated relative to the (best) sequential STL implementation.}
\label{tab:sort_meret_avg}
\end{table*}

\begin{table*}[t!]
\begin{center}
\begin{tabular}{|lr|rr|rrr|rrr||rr|}
\hline
Type & Size & Seq/STL & SeqQS & Fork & SU & Randfork & Cilk & SU & Cilk\_sample & MMPar & SU \\
\hline
 & 10000000 & 1.222 & 1.341 & 0.292 & 4.2 & 1.200 & 0.192 & 6.4 & 0.197 & 0.187 & 6.6 \\
 & 100000000 & 13.232 & 13.492 & 2.585 & 5.1 & 8.442 & 2.252 & 5.9 & 2.073 & 1.081 & 12.2 \\
Random & 1000000000 & 131.266 & 144.100 & 24.698 & 5.3 & 41.073 & 23.345 & 5.6 & 23.205 & 11.121 & 11.8 \\
 & 8388607 & 1.016 & 1.113 & 0.268 & 3.8 & 1.126 & 0.168 & 6.0 & 0.157 & 0.144 & 7.1 \\
 & 33554431 & 4.401 & 4.745 & 0.775 & 5.7 & 3.890 & 0.581 & 7.6 & 0.626 & 0.473 & 9.3 \\
 & 134217727 & 17.537 & 18.405 & 3.249 & 5.4 & 8.504 & 2.390 & 7.3 & 2.366 & 1.513 & 11.6 \\
\hline
 & 10000000 & 1.234 & 1.331 & 0.255 & 4.8 & 1.481 & 0.171 & 7.2 & 0.171 & 0.157 & 7.9 \\
 & 100000000 & 11.861 & 12.945 & 2.454 & 4.8 & 9.870 & 2.372 & 5.0 & 2.449 & 1.343 & 8.8 \\
Gauss & 1000000000 & 119.297 & 129.859 & 22.149 & 5.4 & 87.425 & 23.750 & 5.0 & 23.345 & 13.962 & 8.5 \\
 & 8388607 & 1.027 & 1.105 & 0.227 & 4.5 & 1.045 & 0.166 & 6.2 & 0.166 & 0.155 & 6.6 \\
 & 33554431 & 4.407 & 4.197 & 0.837 & 5.3 & 4.592 & 0.713 & 6.2 & 0.672 & 0.483 & 9.1 \\
 & 134217727 & 15.824 & 17.178 & 2.704 & 5.9 & 13.898 & 2.452 & 6.5 & 2.445 & 1.679 & 9.4 \\
\hline
 & 10000000 & 1.131 & 1.229 & 0.213 & 5.3 & 1.386 & 0.129 & 8.8 & 0.138 & 0.163 & 6.9 \\
 & 100000000 & 12.350 & 12.771 & 1.777 & 6.9 & 10.155 & 1.018 & 12.1 & 1.056 & 1.330 & 9.3 \\
 & 1000000000 & 122.627 & 135.454 & 18.904 & 6.5 & 92.137 & 15.295 & 8.0 & 17.066 & 14.109 & 8.7 \\
Buckets & 8388607 & 0.927 & 1.010 & 0.171 & 5.4 & 1.104 & 0.070 & 13.3 & 0.071 & 0.139 & 6.7 \\
 & 33554431 & 4.109 & 4.493 & 0.621 & 6.6 & 4.041 & 0.490 & 8.4 & 0.525 & 0.471 & 8.7 \\
 & 134217727 & 16.429 & 17.091 & 1.960 & 8.4 & 12.782 & 1.780 & 9.2 & 1.903 & 1.639 & 10.0 \\
\hline
 & 10000000 & 1.141 & 1.287 & 0.247 & 4.6 & 0.840 & 0.378 & 3.0 & 0.413 & 0.172 & 6.6 \\
 & 100000000 & 12.150 & 12.460 & 1.574 & 7.7 & 10.318 & 4.007 & 3.0 & 4.173 & 1.309 & 9.3 \\
Staggered & 1000000000 & 115.845 & 131.236 & 19.672 & 5.9 & 78.962 & 77.297 & 1.5 & 81.088 & 17.095 & 6.8 \\
 & 8388607 & 0.963 & 1.126 & 0.322 & 3.0 & 0.934 & 0.360 & 2.7 & 0.370 & 0.161 & 6.0 \\
 & 33554431 & 4.111 & 4.512 & 0.569 & 7.2 & 2.774 & 0.938 & 4.4 & 1.044 & 0.452 & 9.1 \\
 & 134217727 & 16.217 & 16.838 & 2.230 & 7.3 & 12.705 & 2.056 & 7.9 & 2.127 & 1.705 & 9.5 \\

\hline
\end{tabular}
\end{center}
\caption{Quicksort on the 32-core Intel Nehalem EX system. 
Best (minimum) running time over 10 runs in seconds.
Speedup is calculated relative to the (best) sequential STL implementation.}
\label{tab:sort_meret_min}
\end{table*}

\begin{table*}[t!]
\begin{center}
\begin{tabular}{|lr|rr|rrr||rr|}
\hline
Type & Size & Seq/STL & SeqQS & Fork & SU & Randfork & MMPar & SU \\
\hline
 & 10000000 & 4.541 & 5.449 & 2.128 & 2.1 & 5.036 & 1.464 & 3.1 \\
 & 100000000 & 54.208 & 64.659 & 14.672 & 3.7 & 38.660 & 6.385 & 8.5 \\
Random & 8388607 & 3.718 & 4.441 & 1.509 & 2.5 & 4.548 & 1.094 & 3.4 \\
 & 33554431 & 16.427 & 20.167 & 5.189 & 3.2 & 17.693 & 3.502 & 4.7 \\
 & 134217727 & 75.126 & 86.858 & 16.198 & 4.6 & 45.664 & 10.849 & 6.9 \\
\hline
 & 10000000 & 4.474 & 5.237 & 1.766 & 2.5 & 5.337 & 1.267 & 3.5 \\
 & 100000000 & 52.630 & 62.650 & 13.144 & 4.0 & 37.754 & 5.235 & 10.1 \\
Gauss & 8388607 & 3.552 & 4.545 & 1.578 & 2.3 & 4.094 & 1.149 & 3.1 \\
 & 33554431 & 16.590 & 19.514 & 5.481 & 3.0 & 14.815 & 3.344 & 5.0 \\
 & 134217727 & 72.759 & 90.817 & 23.120 & 3.1 & 56.062 & 9.452 & 7.7 \\
\hline
 & 10000000 & 4.787 & 5.728 & 2.288 & 2.1 & 5.616 & 1.412 & 3.4 \\
 & 100000000 & 56.710 & 67.763 & 16.825 & 3.4 & 41.877 & 7.653 & 7.4 \\
Buckets & 8388607 & 3.807 & 4.516 & 1.439 & 2.6 & 4.404 & 1.220 & 3.1 \\
 & 33554431 & 17.371 & 20.607 & 5.487 & 3.2 & 16.907 & 3.335 & 5.2 \\
 & 134217727 & 76.133 & 91.296 & 21.056 & 3.6 & 68.279 & 11.717 & 6.5 \\
\hline
 & 10000000 & 4.315 & 7.052 & 3.538 & 1.2 & 6.790 & 2.021 & 2.1 \\
 & 100000000 & 52.795 & 79.495 & 27.864 & 1.9 & 50.690 & 8.334 & 6.3 \\
Staggered & 8388607 & 3.570 & 5.376 & 2.037 & 1.8 & 5.439 & 1.438 & 2.5 \\
 & 33554431 & 16.762 & 21.383 & 5.872 & 2.9 & 17.774 & 3.488 & 4.8 \\
 & 134217727 & 71.398 & 102.328 & 31.826 & 2.2 & 56.209 & 8.327 & 8.6 \\
\hline
\end{tabular}
\end{center}
\caption{Quicksort on the 16-core Sun T2+ system running with 32 threads. 
Average running times over 10 repetitions in seconds.
Speedup is calculated relative to the (best) sequential STL implementation.}
\label{tab:sort_t2_32_avg}
\end{table*}

\begin{table*}[t!]
\begin{center}
\begin{tabular}{|lr|rr|rrr||rr|}
\hline
Type & Size & Seq/STL & SeqQS & Fork & SU & Randfork & MMPar & SU \\
\hline
 & 10000000 & 4.526 & 5.440 & 2.025 & 2.2 & 3.031 & 1.252 & 3.6 \\
 & 100000000 & 53.822 & 64.124 & 13.802 & 3.9 & 20.924 & 4.996 & 10.8 \\
Random & 8388607 & 3.698 & 4.418 & 1.355 & 2.7 & 3.055 & 0.753 & 4.9 \\
 & 33554431 & 16.381 & 20.112 & 4.972 & 3.3 & 9.137 & 2.399 & 6.8 \\
 & 134217727 & 74.520 & 86.550 & 15.444 & 4.8 & 33.778 & 8.263 & 9.0 \\
\hline
 & 10000000 & 4.433 & 5.222 & 1.565 & 2.8 & 4.171 & 1.127 & 3.9 \\
 & 100000000 & 52.613 & 62.621 & 12.395 & 4.2 & 18.303 & 4.021 & 13.1 \\
Gauss & 8388607 & 3.543 & 4.532 & 1.427 & 2.5 & 2.881 & 0.976 & 3.6 \\
 & 33554431 & 16.575 & 19.503 & 5.116 & 3.2 & 10.862 & 2.686 & 6.2 \\
 & 134217727 & 72.733 & 90.591 & 21.745 & 3.3 & 41.236 & 7.499 & 9.7 \\
\hline
 & 10000000 & 4.772 & 5.712 & 2.139 & 2.2 & 3.138 & 1.182 & 4.0 \\
 & 100000000 & 56.330 & 67.388 & 15.747 & 3.6 & 23.379 & 5.964 & 9.4 \\
Buckets & 8388607 & 3.802 & 4.511 & 1.313 & 2.9 & 2.546 & 1.141 & 3.3 \\
 & 33554431 & 17.350 & 20.469 & 4.917 & 3.5 & 11.350 & 2.799 & 6.2 \\
 & 134217727 & 76.076 & 91.170 & 20.139 & 3.8 & 42.345 & 8.646 & 8.8 \\
\hline
 & 10000000 & 4.278 & 7.003 & 3.424 & 1.2 & 4.547 & 1.642 & 2.6 \\
 & 100000000 & 52.771 & 79.305 & 25.658 & 2.1 & 42.324 & 7.288 & 7.2 \\
Staggered & 8388607 & 3.565 & 5.363 & 1.903 & 1.9 & 2.879 & 1.273 & 2.8 \\
 & 33554431 & 16.726 & 21.325 & 5.579 & 3.0 & 8.850 & 2.535 & 6.6 \\
 & 134217727 & 71.388 & 102.194 & 30.732 & 2.3 & 40.431 & 6.936 & 10.3 \\
\hline
\end{tabular}
\end{center}
\caption{Quicksort on the 16-core Sun T2+ system running with 32 threads. 
Best (minimum) running time over 10 runs in seconds.}
\label{tab:sort_t2_32_min}
\end{table*}

\begin{table*}[t!]
\begin{center}
\begin{tabular}{|lr|rr|rrr||rr|}
\hline
Type & Size & Seq/STL & SeqQS & Fork & SU & Randfork & MMPar & SU \\
\hline
 & 10000000 & 4.542 & 5.449 & 2.118 & 2.1 & 5.761 & 1.505 & 3.0 \\
 & 100000000 & 53.877 & 64.226 & 14.608 & 3.7 & 44.514 & 8.583 & 6.3 \\
Random & 8388607 & 3.704 & 4.425 & 1.455 & 2.5 & 4.649 & 1.103 & 3.4 \\
 & 33554431 & 16.426 & 20.168 & 4.827 & 3.4 & 19.653 & 2.669 & 6.2 \\
 & 134217727 & 74.590 & 86.664 & 18.152 & 4.1 & 66.932 & 10.323 & 7.2 \\
\hline
 & 10000000 & 4.439 & 5.230 & 1.589 & 2.8 & 5.378 & 1.370 & 3.2 \\
 & 100000000 & 52.634 & 62.659 & 12.912 & 4.1 & 50.805 & 5.321 & 9.9 \\
Gauss & 8388607 & 3.550 & 4.536 & 1.534 & 2.3 & 5.220 & 1.072 & 3.3 \\
 & 33554431 & 16.584 & 19.630 & 5.163 & 3.2 & 18.954 & 3.212 & 5.2 \\
\hline
 & 10000000 & 4.786 & 5.653 & 2.002 & 2.4 & 5.879 & 1.393 & 3.4 \\
 & 100000000 & 57.969 & 68.505 & 17.470 & 3.3 & 53.243 & 8.226 & 7.0 \\
Buckets & 8388607 & 3.860 & 4.628 & 1.545 & 2.5 & 4.920 & 1.075 & 3.6 \\
 & 33554431 & 17.131 & 20.759 & 5.128 & 3.3 & 20.554 & 3.104 & 5.5 \\
 & 134217727 & 77.244 & 91.977 & 21.168 & 3.6 & 70.394 & 10.868 & 7.1 \\
\hline
 & 10000000 & 4.223 & 10.144 & 7.348 & 0.6 & 12.085 & 2.755 & 1.5 \\
 & 100000000 & 51.521 & 97.713 & 54.925 & 0.9 & 84.106 & 15.196 & 3.4 \\
Staggered & 8388607 & 3.713 & 6.778 & 3.117 & 1.2 & 6.922 & 1.915 & 1.9 \\
 & 33554431 & 16.565 & 27.185 & 9.273 & 1.8 & 21.357 & 5.709 & 2.9 \\
 & 134217727 & 71.417 & 174.611 & 78.126 & 0.9 & 123.443 & 29.019 & 2.5 \\
\hline
\end{tabular}
\end{center}
\caption{Quicksort on the 16-core Sun T2+ system running with 64 threads. 
Average running times over 10 runs in seconds.
Speedup is calculated relative to the (best) sequential STL implementation.}
\label{tab:sort_t2_64_avg}
\end{table*}

\begin{table*}[t!]
\begin{center}
\begin{tabular}{|lr|rr|rrr||rr|}
\hline
Type & Size & Seq/STL & SeqQS & Fork & SU & Randfork & MMPar & SU \\
\hline
 & 10000000 & 4.528 & 5.440 & 1.723 & 2.6 & 4.606 & 1.359 & 3.3 \\
 & 100000000 & 53.850 & 64.167 & 13.290 & 4.1 & 25.881 & 7.329 & 7.3 \\
Random & 8388607 & 3.697 & 4.417 & 1.335 & 2.8 & 3.102 & 1.042 & 3.5 \\
 & 33554431 & 16.382 & 20.113 & 4.504 & 3.6 & 10.952 & 2.306 & 7.1 \\
 & 134217727 & 74.554 & 86.591 & 16.489 & 4.5 & 38.155 & 8.517 & 8.8 \\
\hline
 & 10000000 & 4.432 & 5.222 & 1.475 & 3.0 & 4.149 & 1.252 & 3.5 \\
 & 100000000 & 52.619 & 62.622 & 12.399 & 4.2 & 32.363 & 3.967 & 13.3 \\
Gauss & 8388607 & 3.541 & 4.530 & 1.356 & 2.6 & 4.983 & 0.931 & 3.8 \\
 & 33554431 & 16.557 & 19.497 & 4.545 & 3.6 & 12.489 & 2.747 & 6.0 \\
\hline
 & 10000000 & 4.760 & 5.625 & 1.912 & 2.5 & 4.390 & 1.287 & 3.7 \\
 & 100000000 & 57.729 & 68.233 & 15.612 & 3.7 & 31.106 & 6.472 & 8.9 \\
Buckets & 8388607 & 3.848 & 4.621 & 1.355 & 2.8 & 2.966 & 0.985 & 3.9 \\
 & 33554431 & 17.122 & 20.749 & 4.799 & 3.6 & 16.560 & 2.683 & 6.4 \\
 & 134217727 & 77.210 & 91.730 & 20.117 & 3.8 & 43.873 & 9.867 & 7.8 \\
\hline
 & 10000000 & 4.216 & 10.131 & 7.015 & 0.6 & 11.826 & 2.367 & 1.8 \\
 & 100000000 & 51.499 & 97.481 & 52.498 & 1.0 & 68.436 & 13.338 & 3.9 \\
Staggered & 8388607 & 3.702 & 6.768 & 2.823 & 1.3 & 4.519 & 1.654 & 2.2 \\
 & 33554431 & 16.550 & 27.164 & 8.826 & 1.9 & 11.219 & 4.665 & 3.5 \\
 & 134217727 & 71.394 & 174.580 & 74.032 & 1.0 & 93.938 & 24.522 & 2.9 \\
\hline
\end{tabular}
\end{center}
\caption{Quicksort on the 16-core Sun T2+ system running with 64 threads. 
best (minimum) running time over 10 runs in seconds.
Speedup is calculated relative to the (best) sequential STL implementation.}
\label{tab:sort_t2_64_min}
\end{table*}

To evaluate the mixed-parallelism work-stealer we have implemented the
parallel Quicksort algorithm described in~\cite{TsigasZhang03} with
the variations described above. We compare this implementation to the standard
task-based Quicksort algorithm (Algorithm~\ref{alg:stqsort}). The
standard algorithm sequentially partitions the data and then
recursively sorts both created subsequences in parallel. The
\emph{async} statement we use here creates a task out of the following
function call. The \emph{sync} statement waits for all spawned
tasks. We provide a CUTOFF length at which we switch to a sequential
implementation when the task-creation overhead is higher than the
gains.

\begin{algorithm}
\caption{qsort(data, $n$)}\label{alg:stqsort}
\begin{algorithmic}[1]
\IF {$n \leq $ CUTOFF}
	\RETURN sequential\_sort(data, $n$)
\ELSE
	\STATE pivot $\gets$ partition(data, $n$)
	\STATE async qsort(data, pivot)
	\STATE async qsort(data $+$ pivot $+ 1$, pivot $- n - 1$)
	\STATE sync
\ENDIF
\end{algorithmic}
\end{algorithm}

The problem with this algorithm is that during the starting phase we
start with a single sequence that has to be sorted on a single
processor. Only over time we get enough parallel work to fully utilize
all processor resources. In~\cite{TsigasZhang03} this problem is
solved with a data-parallel partitioning step. It starts off with all
processors partitioning a single array. Then, after partitioning is
complete, the processors are split into two groups, where each group
gets a single subsequence to work on. In the final phase, each
processor has a single subsequence that it can sort locally. To
achieve better load-balancing, a helping scheme similar to
work-stealing is used. Therefore, the last phase can be seen as
similar to the task-based Quicksort algorithm in
Algorithm~\ref{alg:stqsort}.

As classic work-stealing is not able to handle data-parallel tasks,
the implementation of Quicksort with data-parallel partitioning has to
rely on manual scheduling and a manually implemented helping
scheme. Our mixed-parallelism work-stealer fits naturally to this
algorithm, and allows to simplify it. Also, it provides better
balancing if other algorithms are executed at the same time, as both
can use the same scheduler. We modify the Quicksort algorithm to use a
more dynamic scheme, which we present in
Algorithm~\ref{alg:mmqsort}. This mixed-mode parallel Quicksort uses a
data-parallel partitioning step, and then launches two subtasks on the
thread with the local id $0$. We modified the \emph{async} to allow
setting the number of threads required by the given task. In this
example, we delegate the task of choosing a good number of threads to
the procedure getBestNp(n). How it is actually implemented may have a
major influence on performance as the overhead for data-parallel
partitioning is higher than for sequential partitioning, so it should
only be used when either the data is large enough so that the overhead
is negligible or there is too little work to do for sequential
tasks. In our implementation we decided on a policy that each thread
working on parallel partitioning should at least have 128 blocks to
work on. To achieve better balancing, we decided to only allow powers
of two as the number of threads for a task.

If the number of threads required by a newly launched task $np$ equals
1, we switch to the standard task-based implementation from
Algorithm~\ref{alg:stqsort}.

\begin{algorithm}
\caption{mmqsort(data, $n$)}\label{alg:mmqsort}
\begin{algorithmic}[1]
\IF {$np = 1$}
	\RETURN qsort(data, $n$)
\ELSE
	\STATE pivot $\gets$ parallel\_partition(data, $n$)
	\IF {localId $=0$ }
		\STATE async(getBestNp(pivot)) mmqsort(data, pivot)
		\STATE async(getBestNp($(n -$ pivot $- 1)/$)) mmqsort(data $+$ pivot $+ 1$, $n -$ pivot $- 1$)
		\STATE sync
	\ENDIF
\ENDIF
\end{algorithmic}
\end{algorithm}

We now explain how the data-parallel partitioning step works. During
partitioning, the array is split into equally sized, cache-aligned
blocks. (The pivot element should not be inside those blocks.) Each
thread takes one block from each side of the array to be sorted, and
tries to \emph{neutralize} (see~\cite{TsigasZhang03} for the details
of this concept) blocks by swapping elements that are larger than the
pivot and in the left block with elements that are smaller than the
pivot and in the right block. As soon as one of the blocks has been
neutralized, the thread tries to acquire another block from the same
side of the array, until we run out of free blocks.

For the second phase, the paper~\cite{TsigasZhang03} proposes that a
single thread then collects the remaining blocks from all other
threads, and processes them sequentially. We decided to follow a
different approach. In our implementation, any thread that needs to
acquire a block decides whether it wants to be a producer, or a
consumer, depending on its current id and the number of blocks on this
side that have to be processed. Producing threads put their remaining
block, and the current processing position into an exchanger
data-structure, and then exit the computation. Consuming threads
retrieve blocks from the exchanger data-structure and continue to
neutralize blocks. During this execution more and more threads switch
from being a consumer to being a producer, until only thread 0
remains.

The third phase starts, when thread 0 only has blocks from one side
remaining. As we now have a sequential execution, we can use a
variation of the sequential partitioner to partition the rest of the
data. Apart from that, our algorithm still uses fork-join parallelism
for its execution, meaning that for each subsequence to sort, a
separate task is created. For the number of threads assigned for each
subtask, we decided to select the biggest power of two, where each
thread can process at least 128 blocks on average during the
partitioning step (of course limited by the number of hardware
threads). If only one thread would process the array, we switch to the
classic fork-join Quicksort implementation with a sequential
partitioning step.

The classic fork-join Quicksort implementation has been designed to
run on the same scheduler. It contains a sequential partitioning step,
and creates a new task for each of the resulting subsequences. If a
subsequence is smaller than a certain size, we switch to the standard
sequential STL sorting algorithm.

Tunable parameters of the Quicksort algorithm are the following:

\begin{itemize}
\item 
Blocksize for parallel partitioning - The Blocksize for parallel
partitioning should be at least as large as the cache-line size. We
decided on a block-size of 4096. (We sorted 4-byte \emph{int} values.)
\item 
Number of threads for the data-parallel partitioning step - In
our implementation, a thread should be able to process at least 16
blocks on average. We only allow powers of two for the number of
threads. Other values might provide better results.
\item 
Cutoff for task-based Quicksort - We decided to let all
subsequences with less than 512 elements be sorted by STL sort.
\end{itemize}

We did not concentrate on finding the best values for those parameters
(or the tuning parameters of the work-stealing scheduler), therefore
performance might be improved using different values.

We compare the mixed-mode parallel Quicksort to a standard task-based
Quicksort implementation. Both are run on our scheduler. We also
implemented a work-stealing scheduler with random stealing and
executed the standard task-based Quicksort on it.  We were not able to
achieve good performance with this version. It seems that random
work-stealing is much more sensible to tuning-parameters, and requires
some more tricks to work well.  Where possible, we also compared to a
task based implementation implemented in
Cilk~\cite{BlumofeJoergKuszmaulLeisersonRandallZhou96}.

Speed-up is in all cases computed relative to the best available
sequential sort implementation which we take to be the STL sort
function. This is also used in our implementation for subsequences
shorter than 512 elements. In the current version of the STL delivered
with \texttt{gcc}, the Introsort algorithm is used that is based on
Quicksort, but has a better worst-case complexity. For each variant,
we took the average and the best (minimum) result out of 10
measurements. We sorted differently generated sequences of 4-Byte
integers distributed as in the
papers~\cite{HelmanBaderJaja98,TsigasZhang03}, namely uniformly
random, random Gaussian, and Buckets and Staggered.

The implementations have been run on four different systems, namely
\begin{itemize}
\item
a 2-socket Intel Nehalem system, where each CPU has 4 cores 
(Intel Xeon X5550 2.66Ghz, 8MB cache).  
\item
a 8-socket AMD Opteron system, where each CPU has 2 cores
(AMD Opteron 8218 2.6 GHz)
\item
a 4-socket Intel Xeon X7560 system, where each CPU has 8 cores
(Intel Xeon X7560 2.26 GHz, 24MB cache)
\item
a 2-socket Sun UltraSPARC T2+ system, where each CPU has 8 cores and 64 hardware threads
\end{itemize}
The collected results are shown in Tables~\ref{tab:sort_cora_avg},
\ref{tab:sort_cora_min}, \ref{tab:sort_daisy_avg},
\ref{tab:sort_daisy_min}, \ref{tab:sort_meret_avg},
\ref{tab:sort_meret_min}, \ref{tab:sort_t2_32_avg},
\ref{tab:sort_t2_32_min}, \ref{tab:sort_t2_64_avg}, and
\ref{tab:sort_t2_64_min}. In these tables, colums Seq/STL list running
times for the ``best'' sequential implementation available (STL),
while columns SeqQS give the running times for the handwritten
reference Quicksort implementation that uses the same cutoff to switch
to STL sort, as the parallel implementations.  Columns Fork are the
running times with a standard, task-based parallel Quicksort
implementation using our work-stealer (all tasks have thread
requirement 1); here both a deterministic and a randomized variant of
the work-stealer have been used.  Columns Cilk give the running times
using Cilk++, wherever Cilk++ could be run (this was not possible on the
Solaris systems). Cilk\_sample denotes the sample quicksort
implementation provided with the Cilk++ compiler, whereas Cilk is a
handwritten example following the same pattern as the other
implementations, including the cutoff. Finally, columns MMPar are our
mixed-mode parallel algorithm shown as Algorithm~\ref{alg:mmqsort}.

Compared to the task-based Quicksort, our mixed mode implementation on
top to the new work-stealing scheduler improves speed-up, often by a
significant fraction; most notably for the Sun T2+ and the large
32-core Intel system. Randomization in the task-based implementation
was tried but turned out to perform poorly, illustrating again that
tuning is important in getting the best performance from a
work-stealing system.  Compared to Cilk, in most cases performance is
comparable, sometimes better, but for the 8-core Nehalem system Cilk
gives systematically better speed-up. This could be due to the fact
that Cilk is more carefully tuned than our prototype system. On the
32-core Nehalem system we achieve consistently better results than even with
Cilk.  On the Sun T2+ system, low speed-up is achieved with 64
threads, while it is quite competitive for 32 threads.  It seems that
the cores are already well utilized with this algorithm when using
2-way SMT, so that nothing can be gained when using more hardware
threads.

\section{Conclusion}

We showed how to extend standard work-stealing to deal with mixed-mode
task and more tightly coupled data parallel programs, in which
dynamically spawned tasks can have fixed requirements for a number
(larger than one) of threads for their execution. We concentrated on
explaining the basic algorithm, which we termed \emph{work-stealing
with deterministic team-building}, and outlined a number of variations
and tunable parameters. A prototype implementation of a such a
work-stealer was given in C++, and used as the basis for implementing
a parallel Quicksort algorithm. On four different many-core systems
with 8 to 32 cores we showed that speed-up could be
improved from $4.8$ using the standard task-based algorithm to $5.6$
using our mixed-mode Quicksort, with an arguably more natural
implementation than in the classic data-parallel
Quicksort~\cite{TsigasZhang03}. 

In future work we will evaluate further mixed-mode parallel
applications, and continue to improve the work-stealing
implementation, including additional ways of improving processor
utilization in cases where the number of threads per task and the
number of processors is not a power of two. One way to do this might
be to allow tasks that are malleable within certain limits. We also
hope to explore the theoretical properties of work-stealing with
deterministic team-building and to explore bounds on the time that
threads may be idle compared to other mixed-mode scheduling
approaches. Eventually we would like to experiment with the approach within
the overall PEPPHER framework.

\end{document}